\documentclass[12pt,reqno]{amsart}

\usepackage{etex}
\usepackage[utf8]{inputenc}
\usepackage{amsfonts}
\usepackage{amsmath}
\usepackage{amssymb}
\usepackage{amsthm}
\usepackage{mathrsfs}
\usepackage{stmaryrd}
\usepackage{color}
\usepackage[english]{babel}
\usepackage[T1]{fontenc}
\usepackage{url}
\usepackage{graphicx}
\usepackage{hyperref}
\usepackage{caption}
\usepackage{enumerate}
\usepackage{epstopdf}
\usepackage{hyphenat}
\usepackage{float}
\usepackage{indentfirst}
\usepackage[export]{adjustbox}
\usepackage{tikz}
\usepackage{color}
\usepackage[all]{xy}
\usetikzlibrary{matrix,arrows,decorations.pathmorphing}
\usepackage{booktabs}
\usepackage{array}
\usepackage{comment}
\usepackage{paralist}
\usepackage{csquotes}
\newcolumntype{P}[1]{>{\centering\arraybackslash}p{#1}}

\allowdisplaybreaks

\newtheorem{theorem}{Theorem}[section]
\newtheorem{lemma}[theorem]{Lemma}
\newtheorem{proposition}[theorem]{Proposition}
\newtheorem{corollary}[theorem]{Corollary}

\theoremstyle{definition}

\newtheorem{remark}[theorem]{Remark}
\newtheorem{remarks}[theorem]{Remarks}
\newtheorem{notation}[theorem]{Notation}
\newtheorem{example}[theorem]{Example}
\newtheorem{definition}[theorem]{Definition}

\DeclareMathOperator{\inid}{in}

\DeclareMathOperator{\Spec}{Spec}
\DeclareMathOperator{\Weil}{Weil}
\DeclareMathOperator{\reg}{reg}
\DeclareMathOperator{\solvdeg}{sd}
\DeclareMathOperator{\maxgb}{max.GB.deg}
\DeclareMathOperator{\ttop}{top}
\DeclareMathOperator{\sat}{sat}
\DeclareMathOperator{\gl}{GL}
\DeclareMathOperator{\id}{Id}

\newcommand{\FF}{\mathbb{F}}
\newcommand{\HS}{\mathrm{HS}}
\newcommand{\HF}{\mathrm{HF}}
\newcommand{\HP}{\mathrm{HP}}
\newcommand{\dregF}[1]{d_{\mathrm{reg}}(#1)}

\makeatletter
\def\l@subsection{\@tocline{2}{0pt}{2.5pc}{5pc}{}}
\makeatother

\author{}
\address{}
\email{}

\title{The complexity of solving Weil restriction systems}

\author{Alessio Caminata}
\address{Alessio Caminata, Dipartimento di Matematica, Universit\`a di Genova\\ via Dodecaneso 35, 16146, Genova, Italy}
\email{caminata@dima.unige.it}

\author{Michela Ceria} 
\address{Michela Ceria, Dipartimento di Meccanica, Matematica e Management, Politecnico di Bari, via Orabona 4, 70125, Bari, Italy}
\email{michela.ceria@gmail.com}

\author{Elisa Gorla}
\address{Elisa Gorla, Institut de Math\'ematiques, Universit\'e de Neuch\^atel, Rue Emile-Argand 11, CH-2000
	Neuch\^atel, Switzerland}  
\email{elisa.gorla@unine.ch}

\begin{document}

	\thanks{ 
		\\  The first author is supported by the Italian PRIN2020, Grant number 2020355B8Y \enquote{Squarefree Gr\"obner degenerations, special varieties and related topics} and by the European Union within the program NextGenerationEU}

	\maketitle
	
	\begin{abstract}
		The solving degree of a system of multivariate polynomial equations provides an upper bound for the complexity of computing the solutions of  the system via Gr\"obner basis methods. In this paper, we consider polynomial systems that are obtained via Weil restriction of scalars. The latter is an arithmetic construction which, given a finite Galois field extension $k\hookrightarrow K$, associates to a system $\mathcal{F}$ defined over $K$ a system $\Weil(\mathcal{F})$ defined over $k$, in such a way that the solutions of $\mathcal{F}$ over $K$ and those of $\Weil(\mathcal{F})$ over $k$ are in natural bijection.
		In this paper, we find upper bounds for the complexity of solving a polynomial system $\Weil(\mathcal{F})$ obtained via Weil restriction in terms of algebraic invariants of the system $\mathcal{F}$.
	\end{abstract}

	\section{Introduction}
	
	The Weil restriction of scalars is a construction which is of interest mostly within arithmetic geometry and number theory. Given a finite Galois field extension $k\hookrightarrow K$, it allows one to associate an object defined over $k$ to one defined over $K$, with the properties that their rational points are in natural bijection. This object can be for example a quasi-projective variety, or an affine or projective scheme. In the case when we start with an affine or projective scheme defined by a system of polynomial equations $\mathcal{F}$ defined over $K$, then the Weil restriction of scalars associates to $\mathcal{F}$ a system $\Weil(\mathcal{F})$ defined over $k$: the defining equations of the Weil restriction of the original scheme. As we already mentioned, the solutions of $\mathcal{F}$ over $K$ and those of $\Weil(\mathcal{F})$ over $k$ are in natural bijection. 
	
	In this situation, the Weil restriction of scalars is of interest also within cryptography.
	This construction has found applications within elliptic and hyperelliptic curve cryptography, where it is used to compute discrete logarithms, see e.g. \cite{Gaudry}. The Discrete Logarithm Problem is a computational problem of central important in public-key cryptography, as several cryptographic primitives rely on its hardness for their security.
	The Weil restriction is also used in multivariate cryptography, one of the current proposals for building post-quantum resistant cryptographic primitives. In this context, the Weil restriction of scalars is useful in order to construct polynomial systems in such a way that their algebraic structure is disguised, so that the designer of the system is the only one who has at their disposal an efficient algorithm to compute its solutions, see e.g.\cite{Patarin}.
	
	The Weil restriction of scalars is also used, although never explicitly mentioned, in coding theory. For example, one can regard spread codes as the rational points of the Weil restriction of scalars of a Grassmannian of lines with respect to an extension of finite fields, see \cite{MGR}. The algebraic structure of spread codes makes their decoding particularly efficient, see \cite{GMR}.
	
	We are interested in estimating the complexity of solving a system of polynomial equations obtained via Weil restriction of scalars by using Gr\"obner basis methods. It is well-known that computing the reduced lexicographic Gr\"obner basis of a system of polynomial equations allows one to compute the solutions of the system, assuming that they are finitely many and that one can efficiently compute the roots of univariate polynomials. This is the case, e.g., over finite fields. Currently, some of the most efficient families of algorithms for computing a Gr\"obner basis are those based on linear algebra, including \cite{F4,XL,F5}. Their complexity is bounded from above by a known function of an invariant of the system, called the solving degree. In this paper we give upper bounds for the solving degree of the Weil restriction of a system of polynomial equations in terms of algebraic invariants of the original system. This gives an upper bound on the complexity of solving the Weil restriction system. 
	
	The paper is organized as follows. In Section \ref{section preliminaries} we define the solving degree of a polynomial system and recall the definition and some facts on the Weil restriction of scalars. The main result of Section \ref{commalg} is Theorem \ref{thm: Weilrestrictionproperties}, where we compute some algebraic invariants of $\Weil(\mathcal{F})$ in terms of those of $\mathcal{F}$. As a consequence, in Corollary \ref{cor: solvingdegreeweilrestriction} we derive an upper bound for the solving degree of a homogeneous Weil restriction system.
	In Section \ref{inhomog} we derive the desired upper bounds for the solving degree of a (not necessarily homogeneous) Weil restriction system. In particular, in Theorem \ref{cor:solvdegWRinhomog} we give an upper bound for the solving degree of $\Weil(\mathcal{F})$ in terms of the Castelnuovo-Mumford regularity of the system obtained from $\mathcal{F}$ by homogenizing its equations. In Corollary \ref{cor:wfieldeqns} we do the same for a system to which we have added the field equations and in Proposition \ref{prop:dreg_inhomog} we relate the degree of regularity of a Weil restriction system to that of the original system.

	\section{Preliminaries}\label{section preliminaries}
	
	In this section we present the definitions and preliminary results that we rely on in the rest of the paper. In \S~\ref{subsection sd and dreg}, we briefly recall the definitions of solving degree and degree of regularity. We limit ourselves to the basic notions necessary to define these two invariants and we refer to \cite{CG21} for a more detailed exposition. In \S~\ref{subsection Weil restriction}, we recall the construction of Weil restriction and we present an algebraic proof of Weil's Theorem in the affine case.
	
	\subsection{Solving degree and degree of regularity}\label{subsection sd and dreg}
	Let $K$ be a field and let $R=K[x_1,\dots,x_m]$ be a polynomial ring in $m$ variables over $K$, equipped with the degree reverse lexicographic term order. 
	We consider a (not necessarily homogeneous) polynomial system $\mathcal{F}=\{f_1,\dots,f_r\}$ in $R$ . The \emph{linear algebra based algorithms} for solving the system $\mathcal{F}$
	transform the problem of computing a Gr\"obner basis of the ideal generated by $\mathcal{F}$ into one or more instances of Gaussian elimination of Macaulay matrices. These are constructed as follows.
	For any degree $d\in\mathbb{Z}_+$ the \emph{Macaulay matrix} $M_{\leq d}$ of $\mathcal{F}$ has columns indexed by the terms of $R$ of degree $\leq d$, sorted in decreasing order from left to right. The rows of $M_{\leq d}$ are indexed by the polynomials $m_{i,j}f_j$, where $m_{i,j}$ is a term in $R$ such that $\deg(m_{i,j}f_j)\leq d$. The entry $(i,j)$ of $M_{\leq d}$ is the coefficient of the term of column $j$ in the polynomial corresponding to the $i$-th row.
	
	The size of the Macaulay matrices $M_{\leq d}$, hence the computational complexity of computing their reduced row echelon forms, depends on the degree $d$. 
	Therefore, it is important to estimate the largest $d$ of the Macaulay matrices involved in the computation of the Gr\"obner basis. For this reason, Ding and Schmidt \cite{DS13} introduced the concept of solving degree.
	
	\begin{definition}[Solving degree]\label{def-solvingdegree}
		Let $\mathcal{F}$ be a polynomial system in $R$. The \emph{solving degree} of $\mathcal{F}$ (with respect to the degree reverse lexicographic term order) is the least degree $d$ such that Gaussian elimination on the Macaulay matrix $M_{\leq d}$ produces a Gr\"obner basis of $\mathcal{F}$. We denote it by $\solvdeg(\mathcal{F})$. 
	\end{definition}
	
	\begin{remarks} 1) One can define and consider the solving degree with respect to any term order.
		However, it turns out that in practice computations with respect to the degree reverse lexicographic term order are often faster than with respect to any other term order. For this reason, in this paper we will only consider the solving degree with respect to the degree reverse lexicographic term order.
		\par 2) Some variants of the algorithms perform Gaussian elimination on $M_{\leq d}$ and then add to the Macaulay matrix $M_{\leq d}$  the rows corresponding to  polynomials $h\cdot f$, where $h$ is a term and $f$ is a polynomial such that $\deg(f) < d$ and the leading term of $f$ was not the leading term of any row of $M_{\leq d}$ before performing Gaussian elimination. Throughout the paper, we consider the situation when \emph{no extra rows are inserted}. Notice that the solving degree is still an upper bound on the degree in which the algorithms adopting this variation terminate. See also \cite[Remark 6]{CG21} for a more detailed discussion.
	\end{remarks}
	
	\par The definition of solving degree has an algorithmic nature and it is usually difficult to estimate the solving degree of a polynomial system without solving it. 
	So many authors use the \emph{degree of regularity} introduced by Bardet, Faug\`ere, and Salvy \cite{Bar04, BFS04} as a proxy for the solving degree.
	
	\par Let $I$ be a homogeneous ideal of $R$, and let $A=R/I$.
	For an integer $d\geq 0$, we denote by $A_d$ the homogeneous part of degree $d$ of $A$.
	The function $\HF_A(-):\mathbb{N}\rightarrow\mathbb{N}$, $\HF_A(d)=\dim_{k}A_d$ is called \emph{Hilbert function} of $A$.
	The \emph{Hilbert series} of $A$ is defined as $\HS_A(z)=\sum_{j\in\mathbb{N}}\HF_A(j)z^j$.
	It is well known that for large $d$, the Hilbert function of $A$ is a polynomial in $d$ called \emph{Hilbert polynomial} and denoted by $\HP_A(d)$.

	\begin{definition}[Degree of regularity]\label{def-dregFaugere}
		Let $\mathcal{F}=\{f_1,\dots,f_r\}\subseteq R$ be a system of equations and let $\left(\mathcal{F}^{\mathrm{top}}\right)=(f_1^{\mathrm{top}},\dots,f_r^{\mathrm{top}})$ be the ideal of $R$ generated by the homogeneous part of highest degree of $\mathcal{F}$. Assume that  $\left(\mathcal{F}^{\mathrm{top}}\right)_d=R_d$ for $d\gg0$.
		The \emph{degree of regularity} of $\mathcal{F}$ is
		\begin{equation*}
			\dregF{\mathcal{F}}=\min\{d\geq 0\mid (\mathcal{F}^{\mathrm{top}})_d=R_d\}=\min\{d\geq 0\mid \HF_{R/(\mathcal{F}^{\mathrm{top}})}(d)=0\}.
		\end{equation*} 
	\end{definition}
	
	For examples and a discussion on the relation between the solving degree and the degree of regularity of a polynomial system we refer the reader to \cite[\S4.1]{CG21}, \cite[\S4.1]{BDDGMT20}, \cite[Corollary 3.67]{T}, \cite[Theorem 2.1]{ST}, and \cite{CG23}.
	
	In this paper, we use results from \cite{CG21} to produce upper bounds for the solving degree. Some of these results require the assumption that the system is in generic coordinates according to \cite[Definition 1.5]{BS87}. We recall the definition for the convenience of the reader. See also \cite[Remark~5]{CG22} for a discussion about the generic coordinates assumption over finite fields.
	
	\begin{definition}\label{defn:gencoords}
		Let $I\subseteq R$ be a homogeneous ideal with $\dim(R/I)=d\geq 0$. We say that $I$ is \emph{in generic coordinates over $\overline{K}$} if $x_i$ is non-zerodivisor $\textrm{mod}\ (I+(x_m,\ldots,x_{i+1}))^{\sat}$ for all $i=m,m-1,\ldots,m-d+1$, where $(I+(x_m,\ldots,x_{i+1}))^{\sat}$ is the saturation of $I+(x_m,\ldots,x_{i+1})$ with respect to the irrelevant maximal ideal of $R$.
	\end{definition}
	
	\begin{remark}\label{remark:gencoords}
		Let $J\subseteq R$ be a homogeneous ideal with $\dim(R/J)=d\geq 0$. Following the terminology of \cite{BS87}, denote by $U_d(J)$ the set of $d$-tuples of homogeneous linear forms $(h_m,\ldots,h_{m-d+1})\in R_1^d$ such that $h_i$ is non-zerodivisor $\textrm{mod}\ (J+(h_m,\ldots,h_{i+1}))^{\sat}$ for  $i=m,m-1,\ldots,m-d+1$.
		The terminology that we use in Definition \ref{defn:gencoords} is motivated by the observation that, for any given $J$, there is an open set of coordinate changes $$U(J)=\{g\in\gl_n(K)\mid (x_m,\ldots,x_{m-d+1})\in U_d(gJ)\}.$$
		Notice that, depending on $J$ and $K$, $U(J)$ may be the empty set.
		Definition \ref{defn:gencoords} states that $I\subseteq R$ is in generic coordinates over $\overline{K}$ if and only if $\id\in U(I\otimes_K\overline{K})$. In particular, $U(I\otimes_K\overline{K})\subseteq\gl_n(\overline{K})$ is a dense open subset of $\gl_n(\overline{K})$.
	\end{remark}

	\subsection{Weil restriction}\label{subsection Weil restriction}
	Let $k\hookrightarrow K$ be a finite Galois field extension of degree $n$ with Galois group $G=\mathrm{Gal}(K/k)$ and let $\{\alpha_1,\dots,\alpha_n\}$ be a $k$-vector space basis  of $K$.
	We consider two polynomial rings: the polynomial ring $R=K[x_1,\dots,x_m]$ in $m$ variables over $K$ and the polynomial ring  $S =k[x_{i,j}]_{i=1,\dots,m, j =1,\dots,n}$   in $nm$ variables over $k$. We define a $K$-algebra homomorphism $\psi:R\rightarrow S\otimes_{k}K$ via
	\begin{equation}\label{eq:changevariables}
		\psi(x_i) = x_{i,1}\alpha_1+\cdots +x_{i,n}\alpha_n \ \ \forall i=1,\dots,m.
	\end{equation}
	
	\begin{definition}\label{def:Weilrestriction}
		Let $f\in R$ be a polynomial. The {\emph {Weil restriction}} of $f$ is the set of polynomials $\Weil(f)=\{f_1,\dots,f_n\}\subseteq S$ defined by
		\[
		f_1\alpha_1+\cdots+f_n\alpha_n=f(\psi(x_1),\dots,\psi(x_n)) \in S\otimes_{k}K,
		\]
		where the right hand side is the image of $f$ under the map $\psi:R\rightarrow S\otimes_{k}K$ defined in \eqref{eq:changevariables}.
		If $\mathcal{F}=\{g_1,\dots,g_r\}\subseteq R$ is a polynomial system, then its Weil restriction is the system $\Weil(\mathcal{F})=\Weil(g_1)\cup\dots\cup\Weil(g_r)\subseteq S$. 
		If $I=(g_1,\dots,g_r)$ is an ideal of $R$, the Weil restriction of $I$ is the ideal $\mathrm{Weil}(I)$ of $S$ generated by the Weil restrictions of $g_1,\dots,g_r$.
		If $V=\Spec(R/I)$ is an affine scheme over $K$, then the Weil restriction of $V$ is the affine scheme $\mathrm{Weil}(V)=\Spec\left( S/\mathrm{Weil}(I)\right)$ over $k$.
	\end{definition}
	
	The definition of Weil restriction of an ideal and of an affine scheme does not depend on the choice of the generators. In fact, the Weil restriction has a more general functorial interpretation.
	
	\begin{remark}
		Given a quasi-projective scheme $V$ over $K$, one can consider the following contravariant functor
		\[
		\begin{split}
			\mathcal{R}_{K\mid k}(V):\mathrm{Sch}/k&\rightarrow\mathrm{Sets}\\
			T&\mapsto\mathrm{Hom}_{K}(T\times_kK,V).
		\end{split}
		\]
		This functor is representable, i.e., there exists a unique scheme $W$ over $k$ such that $\mathcal{R}_{K\mid k}(V)\cong\mathrm{Hom}_k(-,W)$. The scheme $W$ is called the Weil restriction of $V$ with respect to the extension $k\subseteq K$.
		If $V=\Spec(K[x_1,\dots,x_n]/I)$ is an affine scheme, then its Weil restriction $W$ as defined here coincides with the affine scheme $\mathrm{Weil}(V)$ as defined in Definition~\ref{def:Weilrestriction} (see \cite[\S3, Proposition~2]{Nau99}).
	\end{remark}
	
	\par A well-known result by Weil \cite{W82} states that for any  quasi-projective scheme $V$ over $K$ with Weil restriction $W=\Weil(V)$ there is an isomorphism
	\[
	W\times_kK\cong\prod_{\sigma\in G}V^{\sigma},
	\]
	where $V^{\sigma}$ denotes the conjugate of $V$ via $\sigma\in G$.
	We give an algebraic proof of this fact in the affine case. 
	The methods introduced in this proof will be useful throughout the paper.
	
	\begin{theorem}[Weil]\label{thm:weil}
		Let $k\hookrightarrow K$ be a finite Galois field extension of degree $n$ with Galois group $G$, and let $I\subseteq K[x_1,\dots,x_m]$ be an ideal. Let $S=k[x_{i,j}]_{i=1,\dots,m, j =1,\dots,n}$. Then, there is a $K$-algebra isomorphism
		\[
		\Psi:\bigotimes_{\sigma\in G}\left(K[x_1,\dots,x_m]/I\right)^\sigma\longrightarrow (S/\Weil(I))\otimes_k K.
		\]
	\end{theorem}
	
	\begin{proof}
		First of all, since the tensor product of polynomial rings is again a polynomial ring, we can fix an isomorphism
		\[
		\bigotimes_{\sigma\in G}K[x_1,\dots,x_m]^\sigma\cong K[x_{i,\sigma}]_{ i=1,\dots,m, \ \sigma\in G},
		\]
		where the variables $x_{i,\sigma}$ keep track of the action of the Galois group as follows: $\tau\circ x_{i,\sigma}=x_{i,\tau\sigma}$. Similarly, given a polynomial $F$ with coefficients in $K$, we define an action of $\sigma\in G$ on $F$ by $\sigma(F)=F^\sigma$, where the polynomial $F^\sigma$ is the same as $F$, but with all coefficients changed by the action of $\sigma$. Notice that this can be applied to polynomials in $K[x_1,\dots,x_m]$ and $K[x_{i,j}]_{i=1,\dots,m, j =1,\dots,n}$.
		Under the isomorphism above, the product over $G$ of the ideal $I=(F_1,\dots,F_r)$ corresponds to an ideal generated by polynomials $f_{\sigma,t}=F_t^\sigma(x_{1,\sigma},\dots,x_{m,\sigma})$ for $\sigma\in G$ and $t=1,\dots,r$.
		\par We define the $K$-algebra isomorphism 
		\[
		\begin{split}
			\Psi:K[x_{i,\sigma}]_{ i=1,\dots,m, \ \sigma\in G}&\longrightarrow 
			K[x_{i,j}]_{i=1,\dots,m, j =1,\dots,n}\\
			x_{i,\sigma}&\mapsto \sigma(\alpha_1)x_{i,1}+\cdots+\sigma(\alpha_n)x_{i,n}.
		\end{split}
		\] 
		The map $\Psi$ is $K$-linear and its invertibility follows from \cite[Proposition~7.6.6]{Coh03}, since the associated matrix is block diagonal, with each diagonal block of size $n\times n$ and of the form:
		\begin{equation*}
			\begin{pmatrix}
				\alpha_{1} & \alpha_{2} & \cdots & \alpha_{n} \\
				\sigma_1(\alpha_{1}) & \sigma_1(\alpha_{2}) & \cdots & \sigma_1(\alpha_{n}) \\
				\vdots  & \vdots  & \ddots & \vdots  \\
				\sigma_{n-1}(\alpha_{1}) & \sigma_{n-1}(\alpha_{2}) & \cdots & \sigma_{n-1}(\alpha_{n}) \\
			\end{pmatrix}^T,
		\end{equation*}
		where $G=\{\sigma_0=\mathrm{Id},\sigma_1,\dots,\sigma_{n-1}\}$.
		\par We assume for simplicity that $I=(F)$ is principal, with corresponding Weil restriction $\Weil(I)=(f_1,\dots,f_n)$.
		To conclude the proof we have to show that  $\Psi((f_{\sigma_0},...,f_{\sigma_{n-1}}))=(f_1,...,f_n)$, where the equality is intended as ideals.
		\\
		We know that $f_\sigma=F^\sigma(x_{1,\sigma},...,x_{m,\sigma})$. On the other hand, 
		$$F\left(\sum_{j=1}^n \alpha_j x_{1,j}, ..., \sum_{j=1}^n \alpha_j x_{m,j} \right)=\alpha_1 f_1(x_{i,j})+...+ \alpha_n f_n(x_{i,j})$$
		by definition of Weil restriction.
		By letting $\sigma\in G$ act on both sides of the previous equality we obtain
		\[F^\sigma\left(\sum_{j=1}^n \sigma(\alpha_j) x_{1,j}, ..., \sum_{j=1}^n \sigma(\alpha_j) x_{m,j} \right)=  \sigma(\alpha_1) f_1(x_{i,j})+...+ \sigma( \alpha_n) f_n(x_{i,j}).
		\]
		On the other hand,
		$F^\sigma(\sum_{j=1}^n \sigma(\alpha_j) x_{1,j}, ..., \sum_{j=1}^n \sigma(\alpha_j) x_{m,j} )= \Psi(f_\sigma)$.
		This proves that 
		$ \Psi(f_\sigma)= \sigma(\alpha_1) f_1(x_{i,j})+...+ \sigma( \alpha_n) f_n(x_{i,j}) $, which yields $\Psi((f_{\sigma_0},...,f_{\sigma_{n-1}}))=(f_1,...,f_n)$.
		
		If $I$ is not principal, then it suffices to fix a system of generators for $I$ and repeat the same reasoning for each generator.
	\end{proof}

	\section{A commutative algebra approach to Weil restriction}\label{commalg}
	
	In this section, we study some commutative algebra properties of the Weil restriction of an ideal in a polynomial ring. We begin with two preliminary lemmas.
	
	\begin{lemma}\label{lemma:resolutiontensorproduct}
		Let $A=k[x_1,\dots,x_m]$ and $B=k[y_1,\dots,y_t]$ be two polynomial rings over a field $k$, and let $I\subseteq A$ and $J\subseteq B$ be two homogeneous ideals with corresponding minimal graded free resolutions $\mathbb{F}_\bullet\rightarrow A/I\rightarrow0$ and $\mathbb{G}_\bullet\rightarrow B/J\rightarrow0$. Then the product complex $(\mathbb{F}\otimes \mathbb{G})_\bullet$ is a minimal graded free resolution of $A/I\otimes_k B/J\cong k[x_1,\dots,x_m,y_1,\dots,y_t]/(I+J)$.
	\end{lemma}
	
	\begin{proof}
		If we denote by $\mathbb{F}_\bullet=(F_i, d_{i})$ and $\mathbb{G}_\bullet=(G_j, \delta_{j})$ the two complexes, then the tensor product complex is given by  $(\mathbb{F}\otimes\mathbb{G})_h=\displaystyle\bigoplus_{i+j=h}(F_i\otimes_k G_j)$ with differential maps $\partial_h(f\otimes g)=d_i(f)\otimes g + (-1)^if\otimes \delta_j(g) $ for $f\in F_i$, $g\in G_j$.
		From this, we see immediately that if $\mathbb{F}_\bullet$ and $\mathbb{G}_\bullet$ are minimal, i.e. $d_i(F_i)\subseteq (x_1,\dots,x_m)F_{i-1}$ and $\delta_j(G_j)\subseteq (y_1,\dots,y_t)G_{j-1}$, then $(\mathbb{F}\otimes \mathbb{G})_\bullet$ is also minimal. Then from the K\"unneth formula \cite[Corollary~10.84]{Rot09} we obtain $H_0(\mathbb{F}\otimes \mathbb{G})=H_0(\mathbb{F})\otimes_k H_0(\mathbb{G})=A/I\otimes_k B/J$ and $H_n(\mathbb{F}\otimes \mathbb{G})=\displaystyle\bigoplus_{i+j=n}(H_i(\mathbb{F})\otimes_k H_j(\mathbb{G}))=0$ for any $n>0$.
		So $(\mathbb{F}\otimes \mathbb{G})_\bullet$ is exact and resolves $A/I\otimes_k B/J$.
	\end{proof}
	
	\begin{lemma}\label{lemma HSproduct}
		Let $A$ and $B$ be two finitely generated $\mathbb{N}$-graded algebras over a field $k$.
		Then 
		\[\HS_{A\otimes_kB}(z)=\HS_A(z)\cdot\HS_B(z)\]
	\end{lemma}
	
	\begin{proof}
		We observe that $(A\otimes_kB)_h=\bigoplus_{i+j=h}A_i\otimes_k B_j$.
		Therefore we have
		\[
		\begin{split}
			\HS_A(z)\cdot\HS_B(z)&=\sum_h\sum_{i+j=h}\left(\dim_k A_i\cdot\dim_kB_j\right)z^h\\
			&=\sum_h\left(\sum_{i+j=h}\dim_k (A_i\otimes_k B_j)\right)z^h\\
			&=\sum_h\dim_k\left(\bigoplus_{i+j=h}(A_i\otimes_k B_j)\right)z^h\\
			&=\sum_h\dim_k (A\otimes_kB)_h\, z^h\\
			&=\HS_{A\otimes_kB}(z).
		\end{split}
		\]
	\end{proof}

	\begin{theorem}\label{thm: Weilrestrictionproperties}
		Let $k\hookrightarrow K$ be a finite Galois field extension of degree $n$. Let $I\subseteq R=K[x_1,\dots,x_m]$ be a homogeneous ideal, and let $\Weil(I)\subseteq S =k[x_{i,j}]_{i=1,\dots,m, j =1,\dots,n}$ be its Weil restriction. Then
		\begin{enumerate}[(i)]
			\item $\dim(S/\Weil(I))=n\cdot\dim (R/I)$.
			\item $\mathrm{proj.dim}\left(S/\Weil(I)\right)=n\cdot \mathrm{proj.dim}(R/I)$.
			\item If $R/I$ is Cohen-Macaulay, then $S/\Weil(I)$ is Cohen-Macaulay.
			\item If $R/I$ is a complete intersection, then $S/\Weil(I)$ is a complete intersection.
			\item $\reg\left(\Weil(I)\right)=n\cdot \reg(I)-n+1$.
			\item $\HS_{S/\Weil(I)}(z)=\left(\HS_{R/I}(z)\right)^n$.
			\item $e\left(S/\Weil(I)\right)=e\left(R/I\right)^n$, where $e(-)$ denotes the Hilbert-Samuel multiplicity.
		\end{enumerate}	
	\end{theorem}
	
	\begin{proof}
		Since the above properties are invariant under field extension, we will replace $S/\Weil(I)$ by $S/\Weil(I)\otimes_kK\cong S'/\Weil(I)$ where $S'=S\otimes_kK=K[x_{i,j}]_{i=1,\dots,m, j =1,\dots,n}$.
		By Theorem~\ref{thm:weil}, we have a degree-preserving $K$-algebra isomorphism 
		\begin{equation}\label{eq:Weilisomorphism}
			S'/\Weil(I)\cong \displaystyle\bigotimes_{\sigma\in G}(R/I)^{\sigma},
		\end{equation} 
		where $G=\mathrm{Gal}(K/k)$. This yields \textit{(i)}.
		
		Now, let  $\mathbb{F}_{\bullet}\rightarrow R/I\rightarrow0$ be a minimal graded free resolution of $R/I$ with $\mathbb{F}=(F_i, d_i)$.
		For any $\sigma\in G$, a minimal graded free resolution of $(R/I)^\sigma$ is given by $\mathbb{F}_{\bullet}^{\sigma}\rightarrow (R/I)^\sigma\rightarrow0$, where the free modules of $\mathbb{F}_{\bullet}^{\sigma}$ are the same as those of $\mathbb{F}_{\bullet}$ and the differential maps are twisted by the action of $\sigma$ on the coefficients. 
		Therefore, by Lemma~\ref{lemma:resolutiontensorproduct} we have the following minimal graded free resolution:
		\[
		\bigotimes_{\sigma \in G}\mathbb{F}^{\sigma}_{\bullet}\rightarrow\bigotimes_{\sigma\in G} (R/I)^{\sigma}\rightarrow0.
		\]
		We can use this fact to prove the desired properties.
		First, we observe that if $\mathbb{F}_\bullet$ has length $p=\mathrm{proj.dim}(R/I)$, then the length of $\bigotimes\mathbb{F}^{\sigma}_{\bullet}$ will be $np$ with the last non-zero module being $F_p\otimes\cdots\otimes F_p$.
		This proves \textit{(ii)}.
		
		To prove $\textit{(iii)}$ we recall that, by the Auslander-Buchsbaum formula, a quotient $T/J$ of a polynomial ring $T$ by an ideal $J$ is Cohen-Macaulay if and only if $\mathrm{proj.dim}(T/J)=\mathrm{ht}(J)$.
		From \textit{(ii)} and assuming $R/I$ Cohen-Macaulay, we get
		\[
		\mathrm{proj.dim}\left(S'/\Weil(I)\right)=n\cdot\mathrm{proj.dim}(R/I)=n\cdot\mathrm{ht}(I).
		\]
		On the other hand, we have
		\[
		\begin{split}
			\mathrm{ht}\left(\Weil(I)\right)&=\dim (S')-\dim\left(S'/\Weil(I)\right)\\
			&= n\cdot \dim (R) - n \cdot \dim (R/I)\\
			&= n\cdot \mathrm{ht}(I),
		\end{split}
		\]
		where the equality $\dim(S'/\Weil(I))=n\cdot\dim (R/I)$ follows from $\textit{(i)}$.
		
		To prove \textit{(v)}, recall that for any homogeneous ideal $I\subseteq R$ one has $\reg(I)=\reg(R/I)+1$ and $\reg(R/I)=\max\{a_i-i\}$, where $a_i=\max\{j: R(-j)\text{ is a direct summand of }F_i\}.$ We assume that $\reg(R/I)$ is achieved in the resolution $\mathbb{F}_\bullet$ in homological position $i$, that is, $\reg(R/I)=j-i$ with $F_i=R(-j)\oplus F'_i$, and $F'_i$ is a free $R$-module.
		The free module $F_i\otimes\cdots\otimes F_i$ is a direct summand of $\left(\bigotimes\mathbb{F}^{\sigma}\right)_{ni}$, and contains the $R$-module $R(-j)\otimes\cdots\otimes R(-j)\cong R(-nj)$ as direct summand.
		This yields 
		\[
		\reg(S'/\Weil(I))=\reg\left(\bigotimes_{\sigma\in G}(R/I)^{\sigma}\right)\geq nj-ni=n(j-i)=n\cdot\reg(R/I).
		\]
		To prove the reverse inequality, fix $h\leq \mathrm{proj.dim}(S'/\Weil(I))$, and consider
		\[
		\left(\bigotimes\mathbb{F}^{\sigma}\right)_h=\bigoplus_{i_1+\cdots+i_n=h}F_{i_1}\otimes\cdots\otimes F_{i_n}.
		\]
		The maximum shift in each direct summand $F_{i_1}\otimes\cdots\otimes F_{i_n}$ above is $a_{i_1}+\cdots+a_{i_n}$.
		Thus, we obtain
		\[
		a_{i_1}+\dots+a_{i_n}-h=(a_{i_1}-i_1)+\cdots+(a_{i_n}-i_n)\leq n\cdot \reg(R/I).
		\]
		Since the regularity of $S'/\Weil(I)$ is the maximum of all the above expressions of the form $a_{i_1}+\dots+a_{i_n}-h$ for $i_1+\cdots+i_n=h$ and $h\leq \mathrm{proj.dim}(S'/\Weil(I))$, we obtain
		\[
		\reg(S'/\Weil(I))\leq n\cdot \reg(R/I),
		\]
		which gives the desired equality.
		\par Property $\textit{(vi)}$ follows directly from \eqref{eq:Weilisomorphism} and Lemma~\ref{lemma HSproduct}.
		%
		
		Finally, property \textit{(vii)} follows from $\textit{(vi)}$ and the fact that the Hilbert-Samuel multiplicity of $R/I$ is the evaluation in $1$ of the numerator of the simplified Hilbert series of $R/I$, see e.g. \cite[Section 1]{Val98}.
	\end{proof}
	
	From Theorem~\ref{thm: Weilrestrictionproperties} we immediately obtain bounds on the solving degree and degree of regularity of the Weil restriction of homogeneous systems of polynomials.
	In the next section we will extend these results to the non-homogeneous case.
	
	\begin{corollary}\label{cor: solvingdegreeweilrestriction}
		Let $k\hookrightarrow K$ be a finite Galois field extension of degree $n$.  Let $\mathcal{F}\subseteq R=K[x_1,\dots,x_m]$ be a system of homogeneous polynomials with Weil restriction $\Weil\left(\mathcal{F}\right)\subseteq S=k[x_{i,j}]_{i=1,\dots,m, j =1,\dots,n}$.
		\begin{enumerate}
			\item If $\Weil(\mathcal{F})$ is in generic coordinates over $\overline{k}$, then
			\[
			\solvdeg\left(\Weil(\mathcal{F})\right)\leq n\cdot\reg(\mathcal{F})-n+1.
			\]
			\item Assume that $(\mathcal{F})_d=R_d$ for $d\gg 0$. Then \[d_{\reg}\left(\Weil(\mathcal{F})\right)=n\cdot d_{\reg}(\mathcal{F})-n+1.\]
		\end{enumerate}
	\end{corollary}
	
	\begin{proof}
		\textit{(1)} From \cite[Theorem~9]{CG21} we have that 
		\[
		\solvdeg\left(\Weil(\mathcal{F})\right)\leq \reg(\Weil(\mathcal{F})).
		\]
		Now, by Theorem~\ref{thm: Weilrestrictionproperties} we get
		\[
		\reg(\Weil(\mathcal{F}))=n\cdot\reg(\mathcal{F})-n+1
		\]
		as required.
		
		\textit{(2)} Since the system $\mathcal{F}$ is homogeneous and $(\mathcal{F})_d=R_d$ for $d\gg 0$, then $\dregF{\mathcal{F}}=\reg({\mathcal{F}})$. Then the claim follows from Theorem~\ref{thm: Weilrestrictionproperties}.
	\end{proof}

	\section{Solving degree and degree of regularity of Weil restriction systems}\label{inhomog}
	
	We now consider the situation when the system $\mathcal{F}$ is not necessarily homogeneous. In this case, the Weil restriction system
	$\Weil(\mathcal{F})$ is not necessarily homogeneous and we cannot bound its solving degree with the Castelnuovo-Mumford regularity of the corresponding ideal, but we should rather look at the ideal generated by the homogenized system $\Weil(\mathcal{F})^h$. However, we could also swap the two operations. Namely, first homogenize the system $\mathcal{F}$ and then apply Weil restriction to get a new system $\Weil(\mathcal{F}^h)$.
	Clearly, the systems $\Weil(\mathcal{F}^h)$ and $\Weil(\mathcal{F})^h$ are not the same. In fact, they even live in different polynomial rings. However, there is a strict relation between these two systems that we are going to explore in this section. Before that, let us fix the notation.
	
	\begin{notation}\label{notation2}
		Let $k\hookrightarrow K$ be a finite Galois field extension of degree $n$ with a fixed basis $\{\alpha_1=1,\alpha_2,\dots,\alpha_n\}$ of $K$ over $k$.
		Let $R=K[x_1,\dots,x_m]$ be a polynomial ring, and let
		$S =k[x_{i,j}]_{i=1,\dots,m, j =1,\dots,n}$
		be the polynomial ring in $nm$ variables over $k$.
		We consider a system of polynomials $\mathcal{F}=\{f_1,\dots,f_r\}\subseteq R$ not all homogeneous and the corresponding homogenized system $\mathcal{F}^h\subseteq R[t]$ obtained by homogenization with respect to a new variable $t$.
		We have two polynomial systems:
		\begin{itemize}
			\item $\Weil(\mathcal{F}^h)\subseteq S[t_1,\dots,t_n]$, obtained by Weil restriction of $\mathcal{F}^h$, where $t_1,\dots,t_n$ are the variables obtained by Weil restriction from $t$, i.e., $\sum \alpha_i t_i=t$.
			\item $\Weil(\mathcal{F})^h\subseteq S[t]$, obtained by homogenization of $\Weil(\mathcal{F})$ with respect to a new variable $t$.
		\end{itemize}
	\end{notation}

	\begin{lemma}\label{lemma Weil relation}
		Let $\mathcal{F}\subseteq R$ be a system of polynomials as in Notation \ref{notation2}. Then
		the equations of $\Weil(\mathcal{F})^h$ are obtained from those of $\Weil(\mathcal{F}^h)$ by setting $t_1=t$ and $t_2=\cdots=t_n=0$.
	\end{lemma}
	
	\begin{proof}
		Let $f\in\mathcal{F}$, let $d=\deg f$ and write $f=\sum_{a=0}^df_a$, where $f_a$ is homogeneous of degree $a$.
		Then $f^h=\sum_{a=0}^df_at^{d-a}$.
		The polynomials $g_1,\dots,g_n$ of the Weil restriction $\Weil(f^h)$ of $f^h$ are obtained from the relation
		\begin{equation}\label{WRfh}
			\begin{split}
				g_1\alpha_1+\cdots+g_n\alpha_n&=f^h\left(\sum_{j=1}^nx_{i,j}\alpha_j,\sum_{j=1}^nt_j\alpha_j:\ i=1,\dots,m\right)\\
				&= \sum_{a=0}^d f_a\left(\sum_{j=1}^nx_{i,j}\alpha_j:\ i=1,\dots,m\right)\left(\sum_{j=1}^nt_j\alpha_j\right)^{d-a}.
			\end{split}
		\end{equation}
		On the other hand, the polynomials $h_1,\dots,h_n$ of the system $\Weil(f)^h$ are obtained from the identity 
		\[
		\begin{split}
			h_1\alpha_1+\cdots+h_n\alpha_n&= \sum_{a=0}^d f_a\left(\sum_{j=1}^nx_{i,j}\alpha_j:\ i=1,\dots,m\right)t^{d-a},
		\end{split}
		\]
		since the homogenization with respect to $t$ and substituting $\sum_{j=1}^nx_{i,j}\alpha_j$ for $x_i$ are commuting operations. 
		Setting $t_1=t$ and $t_2=\cdots=t_n=0$ into (\ref{WRfh}) yields
		\[
		\begin{split}
			g_1(x_{i,j},t,0,\ldots,0)\alpha_1+\cdots+g_n(x_{i,j},t,0,\ldots,0)\alpha_n&=\sum_{a=0}^d f_a\left(\sum_{j=1}^nx_{i,j}\alpha_j:\ i=1,\dots,m\right)t^{d-a}\\
			&=h_1\alpha_1+\cdots+h_n\alpha_n.
		\end{split}
		\]
		Since $g_1,\ldots,g_n,h_1,\ldots,h_n$ have coefficients in $k$ and $\alpha_1,\ldots,\alpha_n$ are $k$-linearly independent, one has $g_\ell(x_{i,j},t,0,\ldots,0)=h_\ell$ for $\ell=1,\ldots,n$.
	\end{proof}
	
	\begin{example}
		Let $\FF_8=\FF_2[\alpha]$ with $\alpha^3=\alpha+1$ and we fix the basis $\{1,\alpha,\alpha^2\}$ of $\FF_8$ over $\FF_2$.
		We consider the system $\mathcal{F}=\{f\}$ where $f=y^2+xy+\alpha x+\alpha^2\in \FF_8[x,y]$.
		Then we have $f^h=y^2+xy+\alpha xt+\alpha^2t^2$ and its Weil restriction is the system $\Weil(f^h)=\{g_1,g_2,g_3\}$ where
		\[
		\begin{split}
			g_1&=y_1^2+x_1y_1+x_2y_3+x_3y_2+x_1t_3+x_2t_2+x_3t_1+x_3t_3+t_3^2,\\
			g_2&=y_3^2+x_1y_2+x_2y_1+x_2y_3+x_3y_2+x_3y_3+x_1t_1+x_1t_3+x_2t_2+x_2t_3+x_3t_1+x_3t_2+x_3t_3+t_2^2,\\
			g_3&=y_2^2+y_3^2+x_1y_3+x_2y_2+x_3y_1+x_3y_3+x_1t_2+x_2t_1+x_2t_3+x_3t_2+x_3t_3+t_1^2+t_2^2+t_3^2.
		\end{split}
		\]
		Here we used the substitution $x=x_1+\alpha x_2+\alpha^2x_3$ and similarly for $y$ and $t$.
		On the other hand, the polynomials $h_1, h_2,h_3$ of the system $\Weil(f)^h$ are
		\[
		\begin{split}
			h_1&=y_1^2+x_1y_1+x_2y_3+x_3y_2+x_3t, \\
			h_2&=y_3^2+x_1y_2+x_2y_1+x_2y_3+x_3y_2+x_3y_3+x_1t+x_3t,\\
			h_3&=y_2^2+y_3^2+x_1y_3+x_2y_2+x_3y_1+x_3y_3+x_2t+t^2.
		\end{split}
		\]
		It is easy to check that $h_1,h_2,h_3$ are obtained from $g_1,g_2,g_3$ by setting $t_1=t$ and $t_2=t_3=0$.
	\end{example}

	\begin{lemma}\label{lemma t1..tnregular}
		Let $\mathcal{F}\subseteq R$ be a system of polynomials as in Notation~\ref{notation2}. Then $t$ is non zerodivisor in $R[t]/(\mathcal{F}^h)$ if and only if $t_1,\dots,t_n$ are a regular sequence in $S[t_1,\dots,t_n]/(\Weil(\mathcal{F}^h))$.
	\end{lemma}
	
	\begin{proof}
		First, we recall the following fact on Hilbert series, shown in \cite[Proposition 1]{Par10}.
		Given a finitely generated $\mathbb{N}$-graded algebra $A$ over a field and elements $y_1,\ldots,y_r\in A_1$, then $\HS_{A/(y_1,\ldots,y_r)}(z)=(1-z)^r\HS_A(z)$ if and only if $y_1,\ldots,y_r$ is a regular sequence in $A$.
		Hence $t$ is non zerodivisor modulo $(\mathcal{F}^h)$ if and only if
		\[
		\HS_{R/(\mathcal{F}^h\cup\{t\})}(z)=(1-z)\cdot\HS_{R/(\mathcal{F}^h)}(z).
		\]
		Moreover, $t_1,\dots,t_n$ are a regular sequence in $S[t_1,\dots,t_n]/(\Weil(\mathcal{F}^h))$ if and only if 
		\[
		\HS_{S[t_1,\dots,t_n]/((\Weil(\mathcal{F}^h))+(t_1,\ldots,t_n))}(z)=(1-z)^n\cdot\HS_{S[t_1,\dots,t_n]/(\Weil(\mathcal{F}^h))}(z).
		\]
		One also has
		\[
		\HS_{S[t_1,\dots,t_n]/((\Weil(\mathcal{F}^h))+(t_1,\dots,t_n))}(z) = \HS_{S[t]/(\Weil(\mathcal{F}^h\cup\{t\}))}(z)= \left(\HS_{R[t]/(\mathcal{F}^h\cup\{t\})}(z)\right)^n
		\]
		where we used Theorem~\ref{thm: Weilrestrictionproperties} for the second equality, and the fact that $\Weil(\mathcal{F}^h\cup\{t\})=\Weil(\mathcal{F}^h)\cup\{t_1,\dots,t_n\}$ for the first equality. Moreover
		\[
		\HS_{S[t_1,\dots,t_n]/(\Weil(\mathcal{F}^h))}(z) =\left(\HS_{R[t]/(\mathcal{F}^h)}(z)\right)^n.
		\]
		again by Theorem~\ref{thm: Weilrestrictionproperties}.
		It follows that $t$ is non zerodivisor modulo $(\mathcal{F}^h)$ if and only if $t_1,\dots,t_n$ are a regular sequence in $S[t_1,\dots,t_n]/(\Weil(\mathcal{F}^h))$.
	\end{proof}
	
	\begin{theorem}\label{thm:prosecco}
		Let $\mathcal{F}\subseteq R$ be a system of polynomials as in Notation~\ref{notation2} and assume that $t$ is non zerodivisor in $R[t]/(\mathcal{F}^h)$. Then
		\[
		\reg(\Weil(\mathcal{F})^h)= n\cdot\reg(\mathcal{F}^h)-n+1.
		\]
	\end{theorem}
	
	\begin{proof}
		First, we observe that, by Lemma~\ref{lemma t1..tnregular}, $t_1,\dots,t_n$ are a regular sequence in $S/(\Weil(\mathcal{F}^h))$, so in particular $t_2,\dots, t_n$ are.
		We recall also that by \cite[Corollary 4.13]{Eis05} if $M$ is a finitely generated module and $x$ is a linear form that is a non zerodivisor on $M$ then $\reg M = \reg M/xM$.
		Now, we can compute
		\[
		\begin{split}
			\reg\left(\Weil(\mathcal{F})^h\right) &= \reg\left(S[t]/(\Weil(\mathcal{F})^h)\right)+1 \\
			&= \reg\left(S[t_1,\dots,t_n]/((\Weil(\mathcal{F}^h))+(t_2,\dots,t_n))\right)+1\\
			&= \reg\left(S[t_1,\dots,t_n]/\Weil(\mathcal{F}^h)\right)+1\\
			&=n\cdot \reg\left(R[t]/(\mathcal{F}^h)\right)+1\\
			&=n\reg(\mathcal{F}^h)-n+1.
		\end{split}
		\]
		where the second equality follows from Lemma~\ref{lemma Weil relation}, the third one from 
		the observation above since $t_2,\dots, t_n$ are a regular sequence modulo $\Weil(\mathcal{F}^h)$, and the fourth equality follows from  Theorem~\ref{thm: Weilrestrictionproperties}.
	\end{proof}

	For a system $\mathcal{H}$ we denote by $\maxgb(\mathcal{H})$ the largest degree of an element in a reduced degree-reverse-lexicographic Gr\"obner basis of $\mathcal{H}$. 
	We recall that by \cite[Remark 7]{CG21} we have $$\maxgb(\mathcal{H})=\solvdeg(\mathcal{H})$$ for any homogeneous system $\mathcal{H}$.

	\begin{lemma}\label{lemma:sdhomog}
		Let $\mathcal{F}\subseteq R$ be a system of polynomials as in Notation~\ref{notation2}. Assume that $t\nmid 0$ in $R[t]/(\mathcal{F}^h)$. Then $$\solvdeg(\Weil(\mathcal{F}^h))=\solvdeg(\Weil(\mathcal{F})^h).$$
	\end{lemma}
	
	\begin{proof}
		Consider the degree reverse lexicographic order with $t_2,\ldots,t_n$ the last variables. By Lemma \ref{lemma t1..tnregular}, $t_2,\ldots,t_n$ is a regular sequence modulo $\Weil(\mathcal{F}^h)$. Therefore, $t_2,\ldots,t_n$ is a regular sequence modulo $\inid(\Weil(\mathcal{F}^h))$. In other words, the variables $t_2,\ldots,t_n$ do not appear in the monomial minimal generators of $\inid(\Weil(\mathcal{F}^h))$. It follows that, if $\mathcal{G}$ is a minimal Gr\"obner basis of $\Weil(\mathcal{F}^h)$, then $\mathcal{G}\cup\{t_2,\ldots,t_n\}$ is a minimal Gr\"obner basis of $\Weil(\mathcal{F}^h)\cup\{t_2,\ldots,t_n\}$. Since the system $\Weil(\mathcal{F})^h$ does not involve $t_2,\ldots,t_n$ and by \cite[Theorem~7]{CG21} and Lemma \ref{lemma Weil relation} we have
		\[
		\begin{split}
			\solvdeg(\Weil(\mathcal{F})^h) &= \maxgb(\Weil(\mathcal{F})^h)=
			\maxgb(\Weil(\mathcal{F})^h\cup\{t_2,\ldots,t_n\}) \\
			&= \maxgb(\Weil(\mathcal{F}^h)\cup\{t_2,\ldots,t_n\})
			= \maxgb(\Weil(\mathcal{F}^h)) \\ &= \solvdeg(\Weil(\mathcal{F}^h)).
		\end{split}
		\]
	\end{proof}
	
	\begin{theorem}\label{cor:solvdegWRinhomog}
		Let $\mathcal{F}\subseteq R$ be a system of polynomials as in Notation~\ref{notation2}. Assume that $t\nmid 0$ in $R[t]/(\mathcal{F}^h)$ and that $\mathcal{F}^h$ has finitely many projective solutions. Then 
		\[
		\solvdeg\left(\Weil(\mathcal{F})\right)\leq n\cdot\reg(\mathcal{F}^h)-n+1.
		\]
	\end{theorem}
	
	\begin{proof}
		Since $\mathcal{F}^h$ has finitely many projective solutions and $t\nmid 0$ in $R[t]/(\mathcal{F}^h)$, then $R[t]/(\mathcal{F}^h)$ is Cohen-Macaulay of Krull dimension one. Therefore, $S[t_1,\ldots,t_n]/(\Weil(\mathcal{F}^h))$ is Cohen-Macaulay of Krull dimension $n$ by Theorem~\ref{thm: Weilrestrictionproperties}. Since $t\nmid 0$ modulo $\mathcal{F}^h$, then $t_n,\ldots,t_1$ is a regular sequence modulo $\Weil(\mathcal{F}^h)$ by Lemma~\ref{lemma t1..tnregular}.
		Therefore $t_n,\ldots,t_1\in U_n(\Weil(\mathcal{F}^h))$, where we follow the notation of Remark~\ref{remark:gencoords} (see also \cite[Definition 1.5]{BS87}). Hence 
		$$\reg(\Weil(\mathcal{F}^h))=\reg(\inid(\Weil(\mathcal{F}^h)))\geq \maxgb(\Weil(\mathcal{F}^h))=\solvdeg(\Weil(\mathcal{F}^h)),$$
		where the first equality follows from~\cite[Theorem 2.4]{BS87} and the second equality from~\cite[Remark~7]{CG21}.
		Moreover $$\solvdeg\left(\Weil(\mathcal{F})\right)\leq\solvdeg\left(\Weil(\mathcal{F})^h\right)= \solvdeg\left(\Weil(\mathcal{F}^h)\right),$$
		where the inequality follows from~\cite[Theorem~7]{CG21} and the equality from Lemma \ref{lemma:sdhomog}. 
		Summarizing
		$$\solvdeg\left(\Weil(\mathcal{F})\right)\leq\solvdeg\left(\Weil(\mathcal{F}^h)\right)\leq \reg\left(\Weil(\mathcal{F}^h)\right)=n\cdot\reg(\mathcal{F}^h)-n+1,$$
		where the last equality follows from Theorem~\ref{thm:prosecco}.
	\end{proof}
	
	In particular, we obtain the following result for a system defined over a finite field and which contains the field equations.
	
	\begin{definition}
		Let $\mathbb{F}_q$ be a finite field of cardinality $q$. The equations $x_1^q-x_1,\dots,x_m^q-x_m\in \mathbb{F}_q[x_1,\dots,x_m]$ are called the {\em field equations} of $\mathbb{F}_q$.
	\end{definition}

	\begin{corollary}\label{cor:gencoords}
		Let $\FF_q\hookrightarrow \FF_{q^n}$ be an extension of finite fields. Let $\mathcal{F}\subseteq \mathbb{F}_{q^n}[x_1,\dots,x_m]$ and assume that $\mathcal{F}$ contains the field equations of $\mathbb{F}_{q^n}$. Assume also that $t\nmid 0$ in $R[t]/(\mathcal{F}^h)$. Then
		\[
		\solvdeg\left(\Weil(\mathcal{F})\right)\leq n\cdot\reg(\mathcal{F}^h)-n+1.
		\]
	\end{corollary}

	Theorem \ref{cor:solvdegWRinhomog} allows us to extends the Macaulay bound to the Weil restriction of a complete intersection. 
	
	\begin{corollary}
		Let $\mathcal{F}\subseteq R$ be a system of polynomials as in Notation~\ref{notation2}.
		Suppose that $(\mathcal{F}^h)\subseteq R[t]$ is a complete intersection of degrees $d_1,\ldots,d_r$ in generic coordinates over $\overline{K}$, then $\Weil(\mathcal{F})^h$ is a complete intersection in generic coordinates over $\overline{k}$. In this case
		\[
		\solvdeg\left(\Weil(\mathcal{F})\right)\leq n\cdot\reg(\mathcal{F}^h)-n+1=n(d_1+\ldots +d_r)-nr+1.\]
	\end{corollary}
	
	\begin{proof}
		
		If $(\mathcal{F}^h)\subseteq R[t]$ is a complete intersection in generic coordinates over $\overline{K}$, then $t,x_m,\ldots,x_{r+1}$ are a regular sequence modulo $(\mathcal{F}^h)$. By Lemma \ref{lemma t1..tnregular}, $t_n,\ldots,t_1,x_{m,n},\ldots,x_{r+1,1}$ are a regular sequence modulo $\Weil(\mathcal{F}^h)$. Moreover, $\Weil(\mathcal{F}^h)$ is a complete intersection of codimension $nr$ by Theorem \ref{thm: Weilrestrictionproperties} and it is in generic coordinates over $\overline{k}$. By Lemma \ref{lemma t1..tnregular} the same holds for $\Weil(\mathcal{F})^h$.
	\end{proof}
	
	Suppose now that $\mathcal{F}$ contains the fields equations of $\FF_{q^n}$. It is easy to check that $\Weil(\mathcal{F})$ also contains the field equations of $\mathbb{F}_{q^n}$. In practice, in order to solve the system $\Weil(\mathcal{F})$ over $\FF_q$, it makes sense to add the field equations of $\mathbb{F}_{q}$ to it. The next proposition allows us to bound the solving degree of the system that we obtain in this way. The explicit bound is given in Corollary \ref{cor:wfieldeqns}.
	
	\begin{proposition}\label{prop:fieldeqn}
		Let $\FF_q\hookrightarrow \FF_{q^n}$ be an extension of finite fields. Let $\mathcal{F}\subseteq  \FF_{q}[x_1,\dots,x_m]$ be a system of polynomials which contains the field equations of $\FF_{q^n}$. Then
		\[
		\solvdeg\left(\mathcal{F}\cup\{x_i^q-x_i\mid i=1,\ldots,m\}\right)\leq \solvdeg(\mathcal{F}^h).
		\]
	\end{proposition}
	
	\begin{proof}
		The homogenized system $\mathcal{F}^h$ contains the homogenizations of the field equations of $\FF_{q^n}$, namely $x_i^{q^n}-x_it^{q^n-1}$ for all $i=1,\ldots,m$.
		Since $x_i^{q}-x_it^{q-1}\mid x_i^{q^n}-x_it^{q^n-1}$, then 
		$$\maxgb(\mathcal{F}^h\cup\{x_i^q-x_it^{q-1}\})\leq \maxgb(\mathcal{F}^h).$$ 
		Since for a homogeneous system $\mathcal{H}$ we know that $\maxgb(\mathcal{H})=\solvdeg(\mathcal{H})$ by \cite[Remark 7]{CG21}, we obtain that 
		$$\solvdeg(\mathcal{F}^h\cup\{x_i^q-x_i t^{q-1}\mid i=1,
		\ldots,m\})\leq\solvdeg(\mathcal{F}^h),$$
		which concludes the proof.
	\end{proof}
	
	\begin{corollary}\label{cor:wfieldeqns}
		Let $\FF_q\hookrightarrow \FF_{q^n}$ be an extension of finite fields. Let $\mathcal{F}\subseteq \FF_{q^n}[x_1,\dots,x_m]$ be a system of polynomials which contains the field equations of $\FF_{q^n}$. Assume that $t\nmid 0$ modulo $\mathcal{F}^h$. Then
		\[
		\solvdeg\left(\Weil(\mathcal{F})\cup\{x_{i,j}^q-x_{i,j}\}\right)\leq n\cdot\reg(\mathcal{F}^h)-n+1.
		\]
	\end{corollary}
	
	\begin{proof}
		
		Combining Proposition \ref{prop:fieldeqn}, Lemma \ref{lemma:sdhomog}, Corollary \ref{cor: solvingdegreeweilrestriction}, and \cite[Theorem 11]{CG21} we obtain
		\[
		\begin{split}
			\solvdeg\left(\Weil(\mathcal{F})\cup\{x_{i,j}^q-x_{i,j}\}\right)&\leq\solvdeg\left(\Weil(\mathcal{F})^h\right)\\& = \solvdeg\left(\Weil(\mathcal{F}^h)\right)\\&\leq n\cdot\reg(\mathcal{F}^h)-n+1.
		\end{split}
		\]
	\end{proof}
	
	Similarly to the solving degree, one can relate the degree of regularity of a system to that of its Weil restriction.
	
	\begin{proposition}\label{prop:dreg_inhomog}
		Let $\mathcal{F}\subseteq R$ be a system of polynomials. Then $$\Weil(\mathcal{F}^{\ttop})=\Weil(\mathcal{F})^{\ttop}.$$
		Moreover, if $(\mathcal{F}^{\ttop})_d=R_d$ for $d\gg 0$, then
		\[
		d_{\reg}\left(\Weil(\mathcal{F})\right)=n\cdot d_{\reg}(\mathcal{F})-n+1.
		\]
	\end{proposition}
	
	\begin{proof}
		The equality $\Weil(\mathcal{F}^{\ttop})=\Weil(\mathcal{F})^{\ttop}$ follows from observing that substituting the homogeneous linear forms (\ref{eq:changevariables}) commutes with taking the top degree part of the polynomial. Since $(\mathcal{F}^{\ttop})_d=R_d$ for $d\gg 0$, then $\Weil(\mathcal{F}^{\ttop})_d=\Weil(\mathcal{F})^{\ttop}=S_d$ for $d\gg 0$ by Theorem \ref{thm: Weilrestrictionproperties}. Now the relation between the degree of regularity of $\mathcal{F}$ and that of its Weil restriction follows from Corollary \ref{cor: solvingdegreeweilrestriction}.
	\end{proof}

\end{document}